\documentclass[preprint,1p,11pt]{IR-Template/ISAS_IR}

\pdfoutput=1

\usepackage[latin1]{inputenc}
\usepackage{calc}
\usepackage{stfloats} 
\usepackage{flushend}
\usepackage{calc}
\usepackage{microtype}
\usepackage{amsmath}
\usepackage{amssymb}
\usepackage{graphicx}
\usepackage{psfrag}
\usepackage{verbatim}
\usepackage{multirow}
\usepackage{algorithm2e}
\usepackage{blkarray}
\usepackage{lmodern}
\usepackage{wrapfig}
\usepackage{paralist}

\date{}

\def\vec#1{\underline{#1}}

\def\mat#1{{\mathbf #1}}

\def\1_2{{\frac{1}{2}}}

\newcommand{\rv}[1]{\ensuremath{\boldsymbol{#1}}}

\DeclareMathOperator{\EVOp}{E}
\newcommand{\EV}[2]{\EVOp_#2\left\{ #1 \right\}}

\def\SND{standard normal distribution}

\def\etavec{\vec{\eta}}

\def\NewN{\mathbb{N}} 
\def\NewR{\mathbb{R}} 
\def\NewZ{\mathbb{Z}} 
\def\NewC{\mathbb{C}}

\newcommand{\GD}{Gaussian density}
\newcommand{\GDs}{Gaussian densities}
\newcommand{\GMD}{Gaussian mixture density}
\newcommand{\GMDs}{Gaussian mixture densities}
\newcommand{\DC}{Dirac component}
\newcommand{\DM}{Dirac mixture}

\def\Eq#1{(\ref{#1})}

\def\Sec#1{Sec.~\ref{#1}}
\def\SubSec#1{Subsec.~\ref{#1}}

\def\Fig#1{Fig.~\ref{#1}}

\def\Def#1{Definition~\ref{#1}}
\def\Lemma#1{Lemma~\ref{#1}}

\sfcode`\0=1001
\sfcode`\1=1001
\sfcode`\2=1001
\sfcode`\3=1001
\sfcode`\4=1001
\sfcode`\5=1001
\sfcode`\6=1001
\sfcode`\7=1001
\sfcode`\8=1001
\sfcode`\9=1001
\sfcode`\.=1001
\sfcode`\?=1001
\sfcode`\!=1001
\sfcode`\:=1001

\newcommand\AdaptiveCapitalization[1]{\ifnum\ifhmode\spacefactor\else2000\fi>1000 \uppercase{#1}\else#1\fi}

\newcommand{\cM}{\AdaptiveCapitalization{c}ircular moment}
\newcommand{\vMd}{\AdaptiveCapitalization{v}on Mises distribution}
\newcommand{\WNd}{\AdaptiveCapitalization{w}rapped Normal distribution}
\newcommand{\pdf}{\AdaptiveCapitalization{p}robability density function}
\newcommand{\cpdf}{\AdaptiveCapitalization{c}ircular probability density function}

\newcommand{\IKD}{\AdaptiveCapitalization{i}nduced kernel density}
\def\DMA{\DM{} approximation}
\def\DMD{\DM{} density}
\def\DMDs{\DM{} densities}
\newcommand{\PWCD}{\AdaptiveCapitalization{p}iecewise constant density}
\newcommand{\PWCDs}{\AdaptiveCapitalization{p}iecewise constant densities}
\newcommand{\HOM}{\AdaptiveCapitalization{h}igher-order moment}
\def\MEDMA{\AdaptiveCapitalization{m}aximum entropy \DMA{}}
\def\LMDMA{Levenberg-Marquardt \DMA{}}

\usepackage{color}

\def\Box#1{{\fboxsep2pt\fbox{$#1$}}}

\newlength\EqLen

\def\ScaleInner#1{%
\settowidth{\EqLen}{#1}
\ifdim\EqLen < \columnwidth%
\begin{equation*}%
\begin{minipage}{\EqLen}#1\end{minipage}%
\end{equation*}%
\else%
\begin{equation*}%
\resizebox{0.99\columnwidth}{!}{\begin{minipage}{\EqLen}#1\end{minipage}}%
\end{equation*}%
\fi%
}%

\def\Scale#1
  {
  \ScaleInner{$ #1 $} 
  }
  
\def\LongVersion#1{}

\def\citep#1{(\cite{#1})}

\usepackage{amsthm}

\newtheorem{theorem}{Theorem}[section]
\newtheorem{lemma}[theorem]{Lemma}

\theoremstyle{definition}
\newtheorem{definition}{Definition}[section]
\newtheorem{example}{Example}[section]
\theoremstyle{remark}
\newtheorem{remark}{Remark}[section]

\begin{document}

\begin{frontmatter}

\title{Truncated Moment Problem for Dirac Mixture Densities\\
with Entropy Regularization
}

\author[isas]{Uwe~D.~Hanebeck}
\ead{uwe.hanebeck@ieee.org}

\address[isas]{Intelligent Sensor-Actuator-Systems Laboratory (ISAS)\\
Institute for Anthropomatics and Robotics\\
Karlsruhe Institute of Technology (KIT), Germany}

\begin{abstract}
We assume that a finite set of moments of a random vector is given. Its underlying density is unknown. An algorithm is proposed for efficiently calculating  \DMDs{} maintaining these moments while providing a homogeneous coverage of the state space.
\end{abstract}

\end{frontmatter}

\section{Introduction}

\begin{wrapfigure}{r}{0.5\textwidth}
  \vspace*{-10mm}
	\begin{center}
		\includegraphics{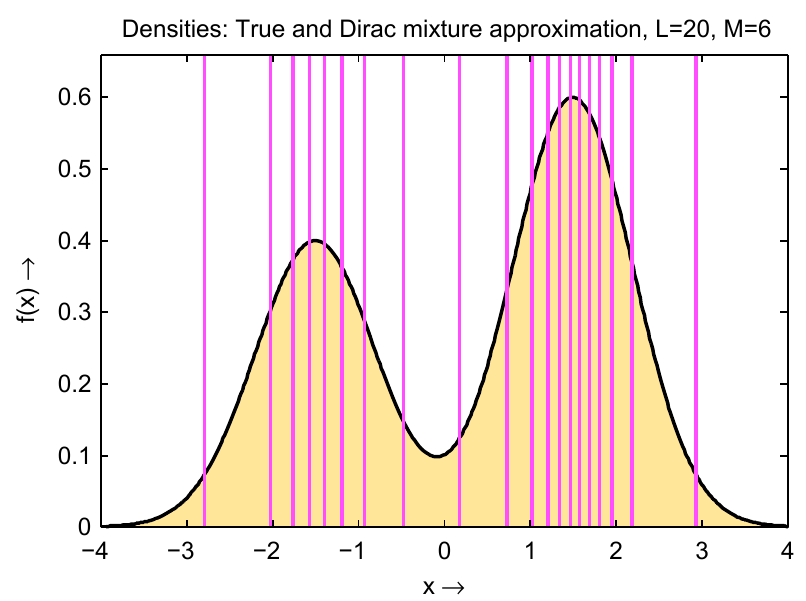}		
  \end{center}
	\vspace*{-5mm}
  \caption{Maximum entropy \DMD{} (purple) with $20$ components and prescribed moments up to order~$6$. The underlying continuous density (yellow) for generating the moments is unknown.}
	\vspace*{-2mm}
\end{wrapfigure}

%
%

We consider a sequence of mappings
\begin{equation}
e_k=\int_{{\cal X}} m_k(\vec{x}) \, \tilde{f}(\vec{x}) \, d \vec{x}
\label{Eq_GeneralMoments}
\end{equation}
from a \pdf{} $\tilde{f}(\vec{x})$ to the so-called moments $e_k$ for $k=0, 1, \ldots$, where ${\cal X}$ is a Polish space. Examples for ${\cal X}$ are the set of real numbers $\NewR$, the $N$-dimensional Euclidean space $\NewR^N$, the unit interval $[0, 1]$, $\NewC$, $\NewC^N$, the unit circle $S^1=\{z \in \NewC: |z|=1\}$, and so forth. Here, we focus on the $N$-dimensional Euclidean space $\NewR^N$.

%
%

We are interested in the inverse problem of deducing the \pdf{} $\tilde{f}(\vec{x})$ from these mapping given the moment sequence. This is called the \emph{moment problem}. Here, we focus on the case that a finite moment sequence $e_k$ for $k=0, 1, \ldots, K$ is given. The problem is then called the \emph{truncated moment problem}.

%
%

Various types of moments can be considered depending upon the functions $m_k(\vec{x})$, $k=0, 1, \ldots, K$. This includes the common power moments and trigonometric moments, which are useful for periodic state spaces such as the unit circle. Here, we focus on power moments.

%
%

We can now ask several fundamental questions such as: Does a density $\tilde{f}(\vec{x})$ exist for the given moment sequence $e_k$, $k=0, 1, \ldots, K$? When a density $\tilde{f}(\vec{x})$ exists, is it uniquely defined by the moment sequence? When it is not uniquely defined, how is the set of densities with the given moment sequence characterized?

%
%

So far, we did not pose restrictions on the \pdf{} $\tilde{f}(\vec{x})$ to be reconstructed from the moment sequence. So, $\tilde{f}(\vec{x})$ can be selected from the space of density functions, which lead to an infinite-dimensional problem. More practical questions on existence, uniqueness, and characterization can be asked, however, when we restrict ourselves to finite-dimensional approximations of the underlying true density $\tilde{f}(\vec{x})$. We consider specific densities $f(\vec{x})$ with a finite-dimensional parametrization, where we have to select a density type with a given structure. In some cases, it is also useful to consider restrictions on their parameter sets. For example, we could ask whether a \GMD{} with three components exists for a given moment sequence and whether it is unique.

%
%

There are many types of parametric densities available such as \GDs{}, \GMDs{}, and exponential densities. In this paper, we consider \DMDs{} $f(\vec{x})$ for approximating the underlying true density $\tilde{f}(\vec{x})$.

%
%

Up to now, the only information available about the underlying true density $\tilde{f}(\vec{x})$ was the moment sequence of length $K+1$. This restricts the number of (independent) parameters of the approximating density $f(\vec{x})$ to be less than or equal to $K+1$. 
%
%
Restricting the number of parameters to $K+1$ is especially problematic as typically the number of available moments is itself limited. For the case of power moments of up to a certain order $M$, the number of moments quickly increases with the number of dimensions $N$ and the order $M$.  Even for a moderate number of dimensions, calculating \HOM s becomes intractable.

%
%

A good coverage of the important regions of the state space with the approximating density $f(\vec{x})$ is mandatory in many applications. For \DMDs{}, this means that we need a large number of components, equivalent to a large number of parameters to be determined.
%
%
When more density parameters than given moments have to be determined, we face an underdetermined inverse problem with an infinite solution set. In that case, additional information about the underlying true density $\tilde{f}(\vec{x})$ such as its support, shape, symmetries, or its smoothness is required. Alternatively, we have to directly impose additional assumptions on the approximating density $f(\vec{x})$. This information can be used to define a regularizer for picking out a single solution with the desired properties.

%
%

An interesting border case is the availability of the full underlying true density $\tilde{f}(\vec{x})$ together with a few of its moments. In that case, we want to find an approximating density $f(\vec{x})$ that maintains the given moments and is in some way as close as possible to the true density.

%
%

A detailed problem formulation is given in the following section including a compact representation of given moments up to a certain order and some words on regularization. \Sec{Sec_StateOfTheArt} gives an overview of the state of the art.

\section{Problem Formulation} \label{Sec_ProbForm}

%
%

A random vector $\rv{\vec{x}} = \begin{bmatrix} \rv{x}_1, \rv{x}_2, \ldots, \rv{x}_N \end{bmatrix}^T \in \NewR^N$ is characterized by a finite set of moments only. The underlying true \pdf{} $\tilde{f}(\vec{x})$ of $\rv{\vec{x}}$ is unknown. The true density $\tilde{f}(\vec{x})$ can be a continuous density or a discrete density on the continuous domain $\NewR^N$.

%
%

Our goal is to represent the unknown \pdf{} $\tilde{f}(\vec{x})$ of the random vector $\rv{\vec{x}}$ by an approximate density $f(\vec{x})$ that has the desired moments. For the approximation, we focus on discrete \pdf s $f(\vec{x})$ on the continuous domain $\NewR^N$. Here, we use a so called \DMD{} $f(\vec{x})$ with $L$ Dirac components given by
\begin{equation}
f(\vec{x}) = \sum_{i=1}^L f_i(\vec{x}) = \sum_{i=1}^Lw_i \cdot \delta(\vec{x} - \hat{\vec{x}}_i) \enspace ,
\label{Eq_DiracMixture}
\end{equation}
with positive weights, i.e., $w_i > 0$ for $i=1, \ldots, L$, that sum up to one, i.e., $\sum_{i=1}^L w_i =1$, and locations $\hat{\vec{x}}_i$ with components $\hat{x}_{ki}$ for dimension $k$ with $k=1,\ldots,N$ and $\hat{\vec{x}}_i \ne \hat{\vec{x}}_j$ for $i=1, \ldots, L$, $j=1, \ldots, L$, $i \ne j$. The locations are collected in a matrix $\hat{\mat{X}}=\begin{bmatrix} \hat{\vec{x}}_1, \hat{\vec{x}}_2, \ldots, \hat{\vec{x}}_L \end{bmatrix} \in \NewR^{N \times L}$.

%
%

Our goal is to systematically find a \DMD{} $f(\vec{x})$ in \Eq{Eq_DiracMixture} that maintains the moments of the true density $\tilde{f}(\vec{x})$ by adjusting its parameters, i.e., its weights $w_i$ and its locations $\hat{\vec{x}}_i$ for $i=1, \ldots, L$. In this paper, we focus on adjusting the locations only. The component weights are all equal. In addition, we assume that a solution exists, i.e., the number of components $L$ is selected to be large enough so that locations exist that fulfill the given moments defined in the next subsection.

%
%

We define a parameter vector $\etavec \in {\cal S} = \NewR^{L \cdot N}$ containing the parameters as
\begin{equation}
\etavec = \begin{bmatrix} \hat{\vec{x}}_1^T, \hat{\vec{x}}_2^T, \ldots, \hat{\vec{x}}_L^T \end{bmatrix}^T
\label{Eq_ParVec}
\end{equation}
and write $f(\vec{x})=f(\vec{x}, \etavec)$.


\subsection{Given Moments}

%
%

We consider power moments for characterizing the random vector $\rv{\vec{x}}$, so we will now specify concrete functions $m_k(\vec{x})$ in \Eq{Eq_GeneralMoments}. For denoting the moment order, we employ a multi-index notation with $\rv{\kappa} = \begin{pmatrix} \kappa_1, \kappa_2, \ldots, \kappa_N \end{pmatrix}$ containing non-negative integer indices for every dimension. We define $|\rv{\kappa}|=\kappa_1+\kappa_2+\ldots+\kappa_N$, $\rv{\kappa}+i=\begin{pmatrix} \kappa_1+i, \kappa_2+i, \ldots, \kappa_N+i \end{pmatrix}$ with $i \in \NewZ$ such that $\rv{\kappa}+i \ge 0$, and $\vec{x}^{\rv{\kappa}} = x_1^{\kappa_1} \cdot x_2^{\kappa_2} \cdots x_N^{\kappa_N}$. For a scalar $c$, the expression $\rv{\kappa} \le c$ is equivalent to $\kappa_k \le c$ for $k=1, 2, \ldots, N$. The power moments of a random vector $\vec{\rv{x}}$ with density $f(\vec{x})$ are given by
%
%
\begin{equation}
e_{\rv{\kappa}} = \int_{\NewR^N} \vec{x}^{\rv{\kappa}} f(\vec{x}) \, d \vec{x}
\label{Eq_Moment}
\end{equation}
for $\kappa \in \NewN_0^N$. For zero-mean random vectors $\vec{\rv{x}}$, the moments coincide with the central moments.

%
%

Moments of a certain order $m$ are given by $e_{\rv{\kappa}}$ for $|\rv{\kappa}| = m$. We define a multi-dimensional matrix $\mat{E}_M$ of moments of up to order $M$ as 
\begin{equation}
\mat{E}_M({\rv{\kappa}}+1) = 
\begin{cases} 
e_{\rv{\kappa}} & |\rv{\kappa}| \le M \\
\text{unspecified} & \text{elsewhere}
\end{cases} \enspace ,
\label{Eq_MomentMatrix}
\end{equation}
with $\rv{\kappa} \le M$ and $e_{\rv{\kappa}}$ from \Eq{Eq_Moment}. We increase the multi-index $\rv{\kappa}$ by $1$, so that the matrix $\mat{E}_M \in \NewR^{(M+1) \times (M+1) \times \ldots  \times (M+1)}$ is indexed from $1$ to $M+1$ in every dimension.

%
%

The matrix $\mat{E}_M$ contains unspecified elements for $|\rv{\kappa}| \le M$ that can either be set to zero in a full matrix or omitted in sparse matrices (when the chosen matrix implementation supports sparse matrices). With ${\cal K}_{NM}$ the set of all \emph{valid} index sequences given by
\begin{equation}
{\cal K}_{NM}=\{\kappa : |\rv{\kappa}| \le M\}
\label{Eq_NumberOfMoments}
\end{equation}
the number of specified elements, i.e., the number of moments of up to order $M$, is defined as
$P_{NM}=|K_{NM}|$ and is given next.

\begin{lemma} \label{Theorem_PNM}
For an $N$-dimensional random vector, the number of moments up to order $M$ is
\begin{equation*}
P_{NM}=\frac{(M+N)!}{M! \, N!}
\end{equation*}
\end{lemma}

\begin{proof}
Elementary.
\end{proof}

For $N=10$ dimensions, the number of moments up to order $M=3$ is $P_{NM}=286$, for $M=5$ already $P_{NM}=3003$.

Of course, there is no need to specify all possible moments for $|\rv{\kappa}| \le M$. In a practical application, there will generally be a lot more unspecified elements.


\subsection{Regularization}

When the length $P_L=L \cdot N$ of the parameter vector $\etavec$ is larger than the number of given moments, the parameters of $f(\vec{x}, \etavec)$ are redundant and a regularizer for $f(\vec{x}, \etavec)$ is required. Here, regularization is performed by selecting the least informative \DM{}, e.g., the one having maximum entropy. As the Shannon entropy for \DMD s is not well defined, we use the entropy of a corresponding \PWCD{}.
%
%
This results in a constrained optimization problem, where the most homogeneous \DMA{} $f(\vec{x}, \etavec)$ is desired that fulfills the given moments \emph{and} maximizes the entropy.

\section{State of the Art} \label{Sec_StateOfTheArt}

%
%

For determining $P_L$ parameters of a \DMD{} in \Eq{Eq_DiracMixture} from a set of $P_{N M}$ moments, we have to distinguish three cases: 
\begin{compactenum}[i)] 
\item $P_L<P_{N M}$, the overdetermined case, i.e., the number of parameters is smaller than the number of moments.
\item $P_L=P_{N M}$, the fully determined case, i.e., the number of parameters is equal to the number of moments.
\item $P_L>P_{N M}$, the underdetermined case, i.e., the number of parameters is larger than the number of moments.
\end{compactenum}


\subsection{$P_L<P_{N M}$, the Overdetermined Case}

%
%
For the overdetermined case, no literature seems to be available. 
%
%
This case is interesting from a theoretical point of view and it will be discussed in more detail later in this paper. 
%
%
From a practical point of view, it makes sense when for some reason the redundancy in the moments can be used to better estimate the parameters of the desired \DMD{}. On the other hand, as discussed above, many \DC s are required to cover the interesting parts of the state space, which leads to a large amount of parameters. It might then be impractical to calculate more moments than parameters.


\subsection{$P_L=P_{N M}$, the Fully Determined Case}

%
%

The fully determined case has been treated a lot in literature in the context of nonlinear Kalman filtering. Moment-based approximations of Gaussian densities are the basis for Linear Regression Kalman Filters (LRKFs) \cite{lefebvre_linear_2005}. Examples are the Unscented Kalman Filter (UKF) in \cite{julier_new_2000} and its scaled version in \cite{julier_scaled_2002}.

%
%

The case of maintaining the first two moments received most attention. Moments of up to second order can be maintained by a \DMD{} with two \DC s per dimension with an optional additional point placed at the mean. This has the huge advantage that the number of components only grows linearly with the number of dimensions.

%
%

Maintaining \HOM s is important for two reasons:
%
%
First, even for Gaussian densities, it makes sense to explicitly consider the \HOM s as the simplest \DM{} (the one with two points per dimension) does not possess the same \HOM s as a \GD{}.
%
%
Second, for non-Gaussian densities \HOM s are essential for capturing asymmetry, multimodality, and so forth.

%
%

Third-order moments are considered in \cite{julier_reduced_2002}. A \DMD{} with $2 N+2$ weighted components is designed that maintains moments of up to second order and, in addition, minimizes the third-order moments. 
%
%
Minimizing the third-order moments is motivated by an underlying \GD{}, but is not useful for general densities, where asymmetries can lead to nonzero third-order moments.

%
%

Under several strong assumptions, \DMDs{} with prescribed higher-order moments have been derived in \cite{tenne_higher_2003}. Assumptions include a placement of components on the coordinate axes only and symmetric densities, so that all odd moments are set to zero. For the actual derivations, \GD s were assumed and the point sets were limited to $4 N +1$ and $6 N+1$ samples with $N$ the number of dimensions.

%
%

\cite{straka_measures_2014} proposes two methods for constructing scalar \DMDs{} with arbitrary first three moments. The first method is a direct approach based on solving for the parameters of a \DMD{} with three weighted components given the first three moments under certain symmetry conditions. Existence of a solution is not guaranteed. The second method is an indirect approach, where a \GMD{} with two components is matched to the given moments, where two degrees of freedom remain to be set by the user. In a second step, two \DMDs{} with three components are calculated matching the first two moments of the individual Gaussian components of the \GMD{}. This results in a final \DMD{} with six components matching the given three moments.

%
%

\DMA{} of \cpdf s analogous to the UKF for linear quantities is introduced in \cite{ACC13_Kurz} for the \vMd{} and the \WNd{}. Based on a closed-form expression for matching the first \cM{}, three \DC s are systematically placed by exploiting symmetry. In \cite{Fusion14_KurzGilitschenski},
a closed-form solution is derived for a symmetric wrapped \DMD{} with five components based on matching the first two circular moments. This \DMA{} of continuous \cpdf s has already been applied to sensor scheduling based on bearings-only measurements \cite{Fusion13_Gilitschenski}. The results have also been used to perform recursive circular filtering for tracking an object constrained to an arbitrary one-dimensional manifold in \cite{IPIN13_Kurz}.


\subsection{$P_L>P_{N M}$, the Underdetermined Case}

%
%
We will now consider the case of more parameters than given moments. 
%
%
As the solution of this inverse problem is not unique, it either requires more information about the underlying density to be reconstructed or assumptions on the desired \DMD{}. In either case, we will perform regularization to guarantee a unique solution with the desired properties.
%
%
We will consider prior information in two different ways:  Either the full density is given or we just know that the underlying true density is smooth.


\subsubsection{Full Density Available}

%
%

When the full underlying true density is given, the most basic problem is to not maintain any moment. Regularization is performed by minimizing a distance between the underlying true density and its \DMA{}. As distances between continuous densities and discrete densities on continuous domains are difficult to define, the densities are typically transformed to a different representation before the distance computation is performed.

%
%

Methods for \DMA{} of scalar continuous densities based on a distance between cumulative distributions with no moment constraint are introduced in
\cite{CDC06_Schrempf-DiracMixt, 
MFI06_Schrempf-CramerMises} 
for a given but arbitrary number of components. An algorithm for sequentially increasing the number of components is given in  
\cite{CDC07_HanebeckSchrempf} 
and applied to recursive nonlinear prediction in
\cite{ACC07_Schrempf-DiracMixt}. 
%

%
%

For arbitrary multi-dimensional Gaussian densities, \DMA s are systematically calculated in \cite{CDC09_HanebeckHuber}. The comparison of densities is performed by comparing probability masses under kernels of arbitrary location and size. For this purpose, the so called Localized Cumulative Distribution (LCD) is introduced in \cite{MFI08_Hanebeck-LCD}. A modified Cram\'{e}r-von Mises distance is then defined based on the LCDs. For the case of \SND s with a subsequent transformation, a more efficient method is given in \cite{ACC13_Gilitschenski}.

%
%

For multidimensional densities with given mean and given variances in every dimension, a method for placing an arbitrary number of \DC s along the coordinate axes is introduced in \cite{IFAC08_Huber}. The multi-dimensional problem is broken down into one-dimensional problems that are solved by minimizing the distance between cumulative distributions given the two moment constraints.

%
%

Multidimensional \DMA s of arbitrary densities with an arbitrary component placement and arbitrary higher-order moment constraints are calculated in \cite{CISS14_Hanebeck}.  Compared to \cite{CDC09_HanebeckHuber}, a faster but suboptimal distance comparison is used. Instead of comparing the probability masses on all scales as in \cite{CDC09_HanebeckHuber}, repulsion kernels are introduced to assemble an \IKD{} and perform the comparison of the true density with its \DMA{}. This method is adapted to Gaussian densities in \cite{Fusion14_Hanebeck}, where a closed-form expression for the distance measure is derived. In addition, a randomized optimization method is employed instead of a quasi-Newton method. This has the advantage of being easier to implement with only a slight decrease in performance.
%
%
An approach similar to the one proposed in \cite{CISS14_Hanebeck} has been derived for \DMA{} of \cpdf s with an arbitrary number of \DC s in \cite{IFAC14_Hanebeck}.


\subsubsection{Smoothness Constraint}

%
%

When it is only known that the underlying density is a smooth continuous density, the first idea that might come to mind is to use an indirect approach. In a first step, a continuous density with the desired moments is found. This can be any smooth parametric density from the exponential family or from a mixture family such as a \GMD{}. In a second step, a \DMA{} of this continuous density is performed. This \DMA{} can be performed with methods discussed before that calculate the \DM{} closest to the given density while simultaneously maintaining the given moments.
 
%
%

The indirect approach has several disadvantages. It is difficult to take care of the given smoothness condition by finding a parametric continuous density first as this includes finding both an appropriate type of density with an appropriate structure and appropriate parameters. This step will most likely introduce unwanted artifacts. In addition, the approximation becomes unnecessarily complicated as we now have to solve two moment problems, a moment problem for the continuous density in the first step and a moment problem for its \DMA{} in the second step.

%
%

We prefer a direct approach, where the smoothness constraint is fulfilled by finding the most homogeneously distibuted \DMA{} under the given moment constraints. To the author's knowledge, no solution to this problem exists for the case of multi-dimensional densities with an arbitrary placement of \DC s.


\subsection{Contribution of this Paper}

%
%

We consider the finite moment problem of calculating parameters of a \DMD{} with a given number of components and prescribed moments.
The true underlying density is unknown.
We focus on redundant problems, where the number of parameters is (much) larger than the number of given moment constraints.
This is an underdetermined problem with an infinite solution set, so that a regularizer is required to exploit redundancy.
Here, we desire a \DMD{} with the most homogeneous coverage under the given moment constraints.
%

%
%

For regularization, the entropy of the \DMD{} could be used. However, the Shannon entropy is not well defined for discrete densities on continuous domains. Here, we propose to use the entropy of the corresponding maximum entropy piecewise constant density approximation. This approximation has a nice interpretation and, for given Dirac components, is given as the solution of a \emph{convex} optimization problem with \emph{linear} inequality constraints. Regularization is now performed by selecting the components of the \DMD{} in such a way that the entropy of this \PWCD{} approximation is maximized.

%
%

The remainder of this paper is structured as follows. The \PWCD{} used for guaranteeing a homogeneous coverage of the final \DMD{} is introduced in \Sec{Sec_PiecewiseConstant}. Calculating the \DMD{} with given moments and homogeneous coverage is discussed in \Sec{Sec_DMA}. Implementation details are given in \Sec{Sec_Implement}. An evaluation is conducted in \Sec{Sec_Evaluation}. Conclusions are drawn in \Sec{Sec_Conclude}.

\section{Piecewise Constant Density Approximation} \label{Sec_PiecewiseConstant}

%
%

In this section, we derive a piecewise constant approximation of the given \DM{}. Its parameters are calculated in such a way that the Shannon entropy is maximized.
%
%
We now define the specific form of piecewise constant density used in this paper.

\begin{definition}[Piecewise constant density]\label{Def_PWCD}
We define a piecewise constant density as a mixture with $L$ components
\begin{equation*}
p(\vec{x}) = \sum_{i=1}^L R(\vec{x}, \hat{\vec{x}}_i, d_i) \enspace ,
\end{equation*}
where each component $R(\vec{x}, \hat{\vec{x}}_i, d_i)$, $i=1, \ldots, L$ is constant within a sphere of radius $d_i$ and given by
\begin{equation*}
R(\vec{x}, \hat{\vec{x}}_i, d_i) = 
\begin{cases} 
h_i & \text{for } \| \vec{x} - \hat{\vec{x}}_i \| \le d_i \\ 
0 & \text{elsewhere}
\end{cases} \enspace ,
\end{equation*}
with $d_i>0$ and the constant heights $h_i$ for each component to be determined. In addition, we desire that the components are disjoint according to
\begin{equation}
d_i + d_j < \| \hat{\vec{x}}_i - \hat{\vec{x}}_j \|
\label{Eq_Constraint}
\end{equation}
holds for all $i=1, \ldots, L$, $j=1, \ldots, L$ with $i \ne j$.
\end{definition}

%
%

\begin{remark}
The diameters $d_i$, $i=1, \ldots, L$ will be collected in a vector $\vec{d}=\begin{bmatrix}d_1, \ldots, d_L \end{bmatrix}^T$.
\end{remark}

\begin{remark}
By exploiting symmetry, \Eq{Eq_Constraint} needs to be checked only for $i=1, \ldots, L-1$, $j=i+1, \ldots, L$. This results in a total of $(L+1)L/2$ inequality constraints.
\end{remark}

%
%

When the piecewise constant density is used as a representation of a given \DM{} with weights $w_i$, $i=1, \ldots, L$, we can calculate appropriate values for the constant heights $h_i$.

\begin{lemma} \label{Lemma_hi}
The constant height for each component $h_i$, $i=1, \ldots, L$ is given by
\begin{equation*}
h_i = \frac{w_i}{V_N(d_i)} \enspace .
\end{equation*}
with $V_N(.)$ the volume of an $N$-dimensional hyper-sphere given by
\begin{equation*}
V_N(d) = \frac{\pi^{\frac{N}{2}}}{\Gamma\left( \frac{N}{2}+1 \right) } \, d^N \enspace .
\end{equation*}

\begin{proof}
As the components of the piecewise constant density $p(\vec{x})$ according to \Def{Def_PWCD} are disjoint, it is possible to treat the individual components separately. We require that the probability mass of the individual component $i$ is equal to the mass of the corresponding Dirac component, which is written as
\begin{equation*}
\int_{\NewR^N} R(\vec{x}, \hat{\vec{x}}_i, d_i) \, d \, \vec{x} = w_i \enspace .
\end{equation*}
By evaluating the left-hand-side, we obtain an integral over the $N$-dimensional hyper-sphere ${\cal S}_N(d_i)$ with radius $d_i$ as
\begin{equation*}
\begin{split}
\int_{{\cal S}_N(d_i)} R(\vec{x}, \hat{\vec{x}}_i, d_i) \, d \, \vec{x} 
& = \int_{{\cal S}_N(d_i)} h_i \, d \, \vec{x} \\
& = h_i \int_{{\cal S}_N(d_i)} \, d \, \vec{x} \\
& = h_i \cdot V_N(d_i)
= w_i \enspace ,
\end{split}
\end{equation*}
which concludes the proof.
\end{proof}
\end{lemma}

%
%

We are now interested in piecewise constant densities that are as homogeneously distributed as possible under certain constraints. For that purpose, we will maximize its Shannon entropy.

\begin{lemma}
The Shannon entropy $h(p)$ of the piecewise constant density $p(.)$ defined in \Def{Def_PWCD} is given by
\begin{equation*}
h(p) = c_N - \sum_{i=1}^L w_i \log\left( \frac{w_i}{d_i^N} \right) \enspace ,
\end{equation*}
where $c_N$ is a constant depending on the number of dimensions $N$.

\begin{proof}
The Shannon entropy of the \PWCD{} $p(\vec{x})$ according to \Def{Def_PWCD} is defined as
\begin{equation*}
h(p) = E \left\{ -\log(p) \right\} = - \int_{\NewR^N} p(\vec{x}) \log\left( p(\vec{x}) \right) \, d \, \vec{x} \enspace .
\end{equation*}
Again, as the components of the \PWCD{} $p(\vec{x})$ are disjoint, we can exchange summation and integration, which gives
\begin{equation*}
h(p) = - \sum_{i=1}^L \int_{\NewR^N} R(\vec{x}, \hat{\vec{x}}_i, d_i) \log\left( R(\vec{x}, \hat{\vec{x}}_i, d_i) \right) \, d \, \vec{x} 
\end{equation*}
or
\begin{equation*}
\begin{split}
h(p) & = - \sum_{i=1}^L \int_{{\cal S}_N}  R(\vec{x}, \hat{\vec{x}}_i, d_i) \log\left( R(\vec{x}, \hat{\vec{x}}_i, d_i) \right) \, d \, \vec{x} \\ 
& = - \sum_{i=1}^L \int_{{\cal S}_N} h_i \log( h_i ) \, d \, \vec{x} \\
& = - \sum_{i=1}^L h_i \log( h_i ) \underbrace{ \int_{{\cal S}_N(d_i)} \, d \, \vec{x} }_{V_N(d_i)} \enspace .
\end{split}
\end{equation*}
With $h_i$ from \Lemma{Lemma_hi}, we obtain
\begin{equation*}
\begin{split}
h(p) & = - \sum_{i=1}^L w_i \log\left( \frac{w_i}{V_N(d_i)} \right) \\
& = \sum_{i=1}^L w_i \left[ \log\left( \frac{\pi^{\frac{N}{2}}}{\Gamma\left( \frac{N}{2}+1 \right) } \right) 
- \log\left( \frac{w_i}{d_i^N} \right)\right] \\
& = \log\left( \frac{\pi^{\frac{N}{2}}}{\Gamma\left( \frac{N}{2}+1 \right) } \right)
- \sum_{i=1}^L w_i \log\left( \frac{w_i}{d_i^N} \right) \enspace ,
\end{split}
\end{equation*}
which gives the desired result.
\end{proof}
\end{lemma}

\begin{lemma} \label{Lemma_MaxEntPWCD}
For a given \DM{} according to \Eq{Eq_DiracMixture} with weights $w_i$ and components $\hat{\vec{x}}_i$, $i=1, \ldots, L$, the corresponding maximum entropy \PWCD{} in \Def{Def_PWCD} is the solution to the optimization problem
\begin{equation*}
\begin{array}{llclll}
\vec{d}=\displaystyle\operatorname*{arg\,max}_{\vec{d}} \left( h \left( p(\vec{\eta}, \vec{d}) \right) \right) & & & & & \\
& \text{ s.t. } & d_i>0 & \text{for} & i=1, \ldots, L & \\[2mm]
& & \multirow{2}{*}{$d_i + d_j < \|\hat{\vec{x}}_i - \hat{\vec{x}}_j \|$} & \multirow{2}{*}{for} & i=1, \ldots, L 
& \multirow{2}{*}{with $i \ne j$ \enspace .} \\
& & & & j=1, \ldots, L &
\end{array}
\end{equation*}
\end{lemma}

\begin{remark}
The optimization (maximization) problem in \Lemma{Lemma_MaxEntPWCD} is characterized by a concave objective function as 
\begin{equation*}
\frac{\partial h \left( p(\vec{\eta}, \vec{d}) \right)}{\partial d_i} = -N\frac{w_i}{d_i}
\end{equation*}
for $i=1, \ldots, L$ and subject to \emph{linear} inequality constraints. It remains to be investigated whether the inequality constraints form a convex set, which would make the optimization problem convex.
\end{remark}

\begin{remark}
In general, the optimization problem in \Lemma{Lemma_MaxEntPWCD} requires a total of $L(L+1)/2$ \emph{linear} inequality constraints: $d_i>0$ for $i=1, \ldots, L$, which gives $L$ inequality constraints, and $d_i + d_j < \| \hat{\vec{x}}_i - \hat{\vec{x}}_j \|$ for $i=1, \ldots, L$, $j=1, \ldots, L$, and $i \ne j$, which together with symmetry gives another $L(L-1)/2$ inequalities. The scalar case $N=1$ is an exception as the positions $\hat{x}_i$, $i=1, \ldots, L$ can be ordered\footnote{W.l.o.g., we can assume distinct positions, i.e., $\hat{x}_i \ne \hat{x}_j$ for $i=1, \ldots, L$, $j=1, \ldots, L$, and $i \ne j$. Otherwise, the number of points would have been reduced accordingly.}. In that case, we have a total of $2 L -1$ \emph{linear} inequality constraints: $d_i>0$ for $i=1, \ldots, L$ and $\hat{x}_i+d_i < \hat{x}_{i+1}-d_{i+1}$ for $i=1, \ldots, L-1$.
\end{remark}

\section{Homogeneous Dirac Mixture Approximation with Given Moments} \label{Sec_DMA}

%
%

Our goal is an algorithm for the efficient calculation of a \DMD{} with given moments and a homogeneous coverage of the state space. 
%
%
We will start with taking a look at the given moments. They will act as constraints for the desired \DMD{}.

%
%

Given moments up to an order $M$ are stored in the matrix $\mat{E}_M$ from \Eq{Eq_MomentMatrix}. These given moments will be denoted as 
$\tilde{\mat{E}}_M$. The actual moments of the \DMD{} in \Eq{Eq_DiracMixture} are just denoted by $\mat{E}_M$, so that our moment constraints can be written as
\begin{equation*}
\mat{E}_M \stackrel{!}{=} \tilde{\mat{E}}_M \enspace ,
\end{equation*}
where the left hand side depends on $\vec{\eta}$, i.e., we have $\mat{E}_M=\mat{E}_M(\vec{\eta})$.

%
%

Depending on the number of given moments compared to the number of parameters of the desired \DMD{}, we have to distinguish between three different cases:
\begin{description}
\item[Case 1:] In the first case, the number of given moments is larger than the number of parameters. Now, the system of nonlinear equations given by
\begin{equation*}
\mat{E}_M(\vec{\eta}) - \tilde{\mat{E}}_M \stackrel{!}{=} 0
\end{equation*}
is overdetermined as it has more equations than unknowns, i.e., number of parameters or length of the parameter vector $\vec{\eta}$. In that case, we can calculate a least-squares solution as
\begin{equation*}
\vec{\eta} = \operatorname*{arg\,max}_{\vec{\eta} \in {\cal S}} \left\| \tilde{\mat{E}}_M - \mat{E}_M \right\|_F^2 \enspace ,
\end{equation*}
where $\|.\|_F$ is the Frobenius norm defined for a matrix $\mat{H} \in \NewR^{M_1 \times M_2}$ with elements $h_{ij}$, $i=1, \ldots, M_1$, $j=1, \ldots, M_2$ as
\begin{equation*}
\|\mat{H}\|_F = \sqrt{\sum_{i=1}^{M_1} \sum_{j=1}^{M_2} |h_{ij}|^2} \enspace .
\end{equation*}
A generalization would be to introduce weighting factors for moments of different orders.
\item[Case 2:] In the second case, the number of given moments is exactly equal to the number of parameters. In this case, we have to solve the system of nonlinear equations given by
\begin{equation*}
\mat{E}_M(\vec{\eta}) - \tilde{\mat{E}}_M \stackrel{!}{=} 0 \enspace .
\end{equation*}
for the parameter vector $\vec{\eta}$. In the case of power moments, the left hand side leads to a system of multi-dimensional higher-order polynomials. Analytic solutions are only available in rare special cases, so that numerical root-finding technqies have to be applied. In addition, the solution is not necessarily unique.
\item[Case 3:] The third and last case is the one that we will pursue further. Here, the number of moment constraints is smaller than the number of parameters, so calculating the desired \DMD{} is underdetermined. Just considering the moment constraints would give an infinite solution set for the desired parameter vector $\vec{\eta}$.
\end{description}

%
%

We will now pursue the third case. In that case, the solution, the parameter vector of the \DMD{} with the desired moments, is underdetermined. Thus, we have redundancy available that can be exploited, so that we can impose additional constraints on the desired \DMD{} to perform a regularization and to finally end up with a unique solution. 

%
%

Here, we want the final \DMD{} to be not more informative as already specified by the given moments: It should be as uninformative as possible within the given moment constraints. The density in a \DM{} is encoded by both the weights and the ``density'' of its components, that is their relative spacing. When we do not want to favor certain regions of the state space in terms of their density, the \DM{} should be as homogeneous as possible in terms of weight distribution and location distribution. Intuitively, for equally weighted components we somehow desire equal distances between neighboring components or equivalently equal free spaces around each component. As this is not very precise and does not hold for unequally weighted components, we need a formal definition of homogeneity.

%
%

As a convenient, effective, and intuitive regularizer we use the corresponding \PWCD{} for a \DMD{} as introduced in \Sec{Sec_PiecewiseConstant}. The most homogenous \DM{} from the infinite solution set then is defined as the one that maximizes the entropy of the corresponding \PWCD{}.

%
%

For performing the optimization, we again couple the \PWCD{} with the \DM{} in such a way that the locations $\hat{\vec{x}}_i$ of the \DM{} are the midpoints of the spherical support for each component of the \PWCD{}. The probability mass $w_i$ of the \DM{} determines the height $h_i$ of each  component of the \PWCD{}. In contrast to \Sec{Sec_PiecewiseConstant}, now both the diameters $d_i$ and the locations $\hat{\vec{x}}_i$ are variables that are simultaneously optimized. The optimization result provides the locations $\hat{\vec{x}}_i$ of the most homogenous \DM{}. The optimal diameters $d_i$ of the maximum entropy \PWCD{} are a by-product and can be used for visualization.

\begin{theorem} \label{Lemma_DMA}
For given moments collected in the moment matrix $\tilde{\mat{E}}_M$, a \DMD{} with the same moments and a corresponding maximum entropy \PWCD{} is the solution to the optimization problem
\begin{equation*}
\begin{array}{llclll}
[\vec{\eta}, \vec{d}]=\displaystyle\operatorname*{arg\,max}_{\vec{\eta}, \vec{d}} \left( h \left( p(\vec{\eta}, \vec{d}) \right) \right) & & & & & \\
& \text{ s.t. } & \mat{E}_M = \tilde{\mat{E}}_M & & & \\[2mm]
& & d_i>0 & \text{for} & i=1, \ldots, L & \\[2mm]
& & \multirow{2}{*}{$d_i + d_j - \|\hat{\vec{x}}_i - \hat{\vec{x}}_j \| < 0$} & \multirow{2}{*}{for} & i=1, \ldots, L 
& \multirow{2}{*}{with $i \ne j$ \enspace .} \\
& & & & j=1, \ldots, L &
\end{array}
\end{equation*}
\end{theorem}

Of course, compared to the sole optimization of the diameters of the \PWCD{} in \Sec{Sec_PiecewiseConstant}, this optimization problem now has nonlinear constraints: The equality constraints for maintaining the desired moments are typically nonlinear. Also the inequality constraints for avoiding collisions of the spherical supports of the \PWCDs{} are now nonlinear as the locations $\hat{\vec{x}}_i$ now are variables.

\begin{remark}
It is important to note that all the required statistics such as the moments are directly calculated from the \DM{}. They are \emph{not} calculated from the corresponding \PWCD{} as this one solely serves regularization purposes.
\end{remark}

\section{Implementation and Complexity} \label{Sec_Implement}

%
%

Given the algorithm derived in the previous section, our goal is the efficient calculation of a \DMD{} with given moments and a homogeneous coverage of the state space. Homogeneity is achieved by a regularizer that picks out the most homogeneous solution from the infinite solution set that we have in case of more parameters than moment constraint.


\subsection{Symmetric Densities} \label{Subsection_Symmetry}

%
%

Often, more information besides the moments is available about the true density $\tilde{f}(\vec{x})$, such as information on its support or on given symmetries. Even when information of this type is unavailable, analogous assumptions could be made about the approximating density $f(\vec{x})$.

%
%

Here, we consider a special case of symmetric densities, \DMDs{} $f(\vec{x})$ that are symmetric with respect to their expected value $\EV{\vec{x}}{f}$. We only have to specify $L$ \emph{master components} $f_1, \ldots, f_L$ and implicitly end up with $2 L$ components, i.e., $f_1, \ldots, f_{2 L}$, where the \emph{master components} $f_1, \ldots, f_L$ control \emph{slave components} $f_{L+1}, \ldots, f_{2 L}$. The slave components are symmetric copies of the master components in the sense, that
\begin{equation*}
\hat{\vec{x}}_{i+L}-\EV{\vec{x}}{f}=-(\hat{\vec{x}}_i-\EV{\vec{x}}{f})
\end{equation*}
holds for $i=1, \ldots, L$. The weights $w_i$ and the diameters $d_i$ are simply copied according to $w_{i+L}=w_i$ and $d_{i+L}=d_i$ for $i=1, \ldots, L$.

%
%

Exploiting symmetries in this form has two major advantages. First, complexity is reduced as only half as many variables have to be optimized (this does not depend upon the number of dimensions). Second, the prescribed expected value is automatically maintained without an explicit constraint.

%
%

\begin{remark}
The master components are not confined to specific parts of the state space. They can be located anywhere, which simplifies the implementation.
\end{remark}


\subsection{Complexity}

%
%

Regularization is performed by maximizing the Shannon entropy of the corresponding \PWCD{}, which requires complying with a number of constraints quadratic in the number of \DC s. We will now take a look at the complexity, especially the number of constraints to maintain. The number of equality constraints (the moment constraints) is prespecified and for a given order $M$ bounded by \Eq{Eq_NumberOfMoments}. For the inequality constraints, we have $L$ linear positivity constraints for the radii $d_i$, $i=1,\ldots, L$ of the \PWCD{} and $(L-1)L/2$ nonlinear collision constraints. The latter ones are the most critical and will be investigated further, where we distinguish the scalar or $1$-dimensional case and the $N$-dimensional case with $N>1$.

%
%

In the scalar case, the number of collision constraints is reduced by maintaining an ordered list of \DC s. Then, only the distances between neighbors have to be considered, which reduces the number of collision constraints from $(L-1)L/2$ to $L-1$.

%
%

For the general multi-dimensional case with $N>1$, the number of collision constraints does not depend upon the dimension. However, it depends quadratically upon the number of components. For a small number of components, say up to $20$, this poses no problem. For more components, however, the number of constraints needs to be reduced.

%
%

Reducing the number of collision constraints while still guaranteeing a correct solution will be pursued in a follow-up paper. The key is to exploit the fact that simpler constraints can be devised to check for collisions of the support spheres onto the coordinate axes. This is much simpler and leads to fewer constraints. Non-colliding projections are a sufficient, but not a necessary, condition for non-colliding supports. Based on this insight, expensive collision constraints have to be considered for far fewer components. Details are given in the conclusions.


\section{Evaluation} \label{Sec_Evaluation}

%
%

We will now demonstrate the performance of the proposed \DMA{} method by some examples. One-dimensional and two-dimensional densities will be considered.

%
%

\begin{remark}
It is important to note, that in all cases where we consider underlying continuous densities, these are only used for generating the moments. They are \emph{not known} to the approximation methods!
\end{remark}

%
%

For comparison, we employ a solver that directly finds a root of the underdetermined system of equations given by the moment constraints. We use a 
Levenberg-Marquardt method for that purpose, that we from now on call \LMDMA{} method. Matlab provides an implementation by calling {\tt fsolve} with the appropriate options. This solver does neither provide unique nor reproducible results.

%
%

In the simulations, both optimization methods, the \LMDMA{} method and the proposed \MEDMA{} method, are initialized with a random parameter vector $\vec{\eta}$ drawn from a standard normal distribution.


\subsection{Examples for the One-dimensional Case}

%
%

We begin with the simplest case of generating moments from a \GD{} that are then used for characterizing a \DMD{}. For a standard normal distribution, we calculate the first two moments $e_1$, $e_2$. Then, we employ the \LMDMA{} method and the \MEDMA{} method to find \DMDs{} with $L=6$ components having exactly these moments. A comparison of densities and distributions is shown in \Fig{Fig_Gaussian_Compare_LM_ME}, where the top two figures show the result of the \LMDMA{}. This is just one representative result, as the solution is not unique and changes for every optimization performed. The bottom row in \Fig{Fig_Gaussian_Compare_LM_ME} shows the result of the \MEDMA{}. Here, the coverage is much more homogeneous. In addition, it becomes clear that it comes close to the underlying Gaussian as a maximum entropy solution is considered and the Gaussian is the continuous density with the highest entropy given a certain variance.
%
%
This convergence becomes even clearer when taking a look at \Fig{Fig_Gaussian_Convergence}. When the number of \DC s increases, in that case to $L=10$ and $L=15$, the generated \DM{} converges to the underlying Gaussian \emph{although this density is not known to the algorithm}.

\begin{figure*}[t]
\includegraphics{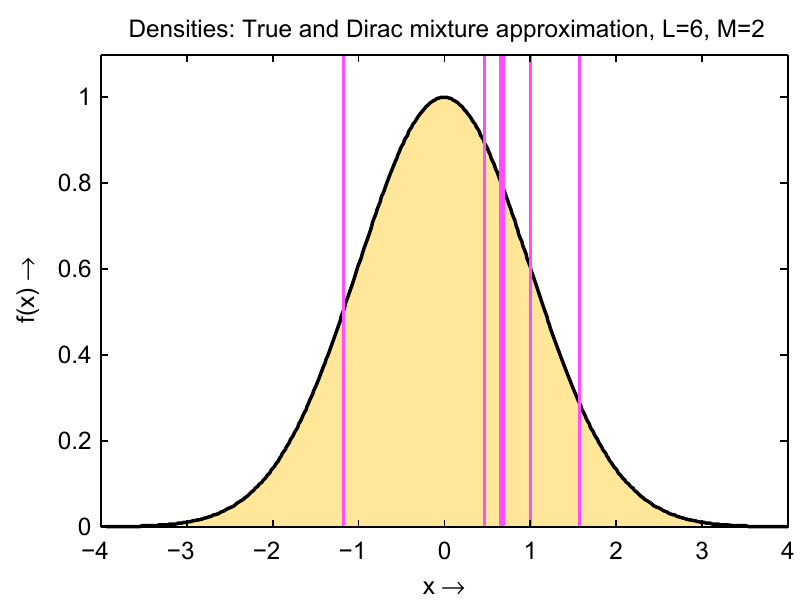}
\hspace*{\fill}
\includegraphics{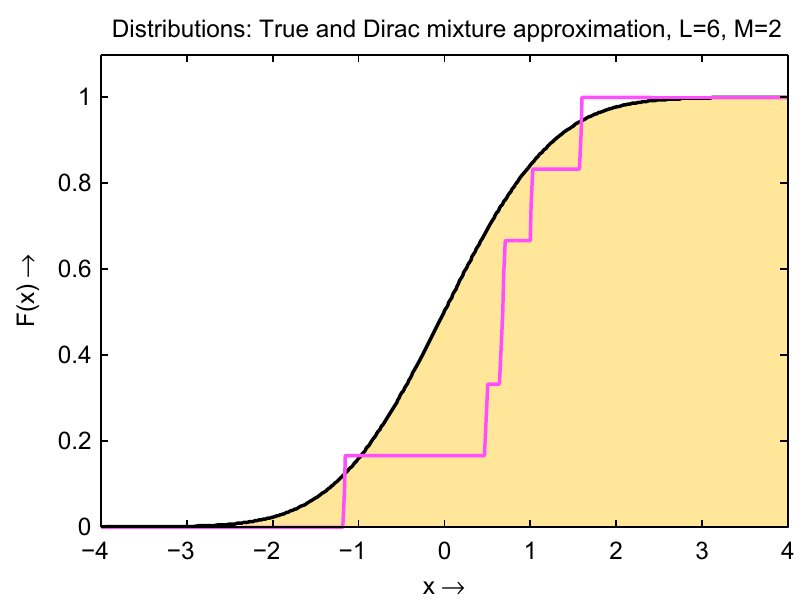}\\
\includegraphics{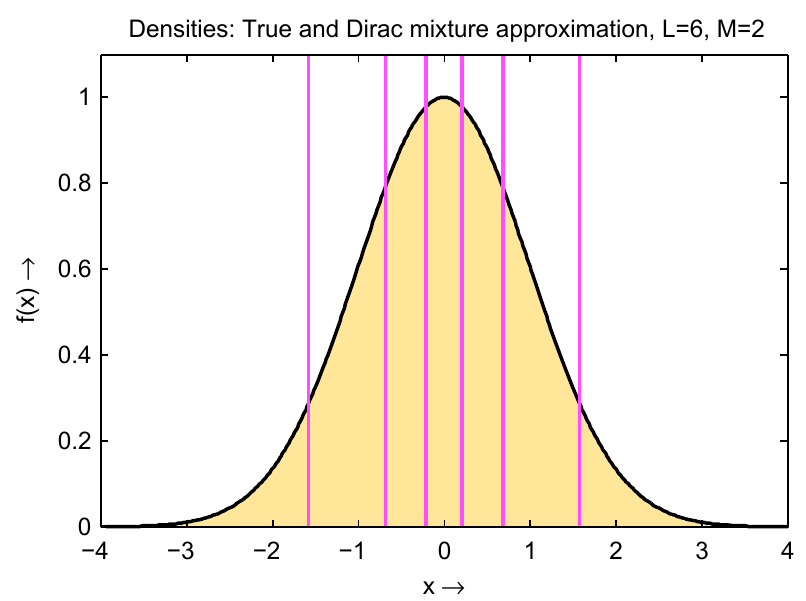}
\hspace*{\fill}
\includegraphics{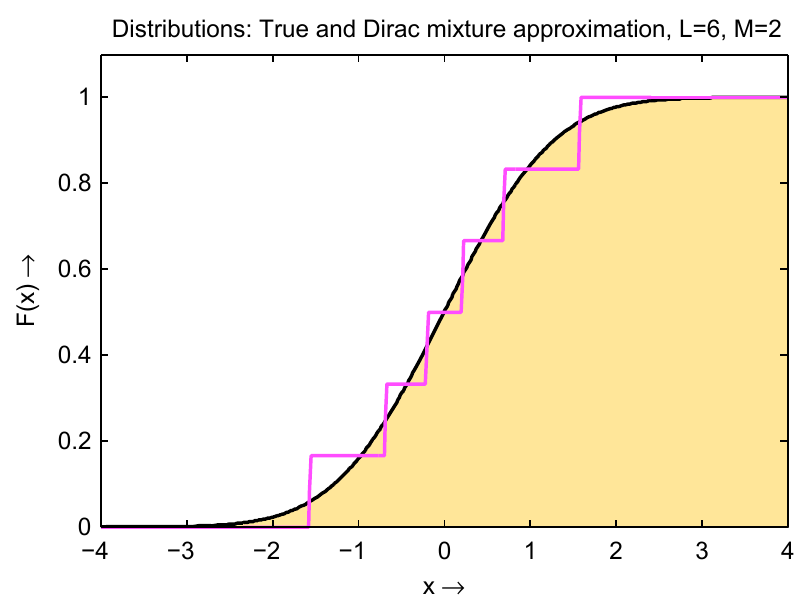}
\caption{Comparison of two \DMA{} methods for an underlying \GD{}. (Top row) Result of solving the underdetermined moment constraints with the \LMDMA{} method. (Bottom row) Result of the proposed \MEDMA{} method. (Left column) Comparison of the densities, where the (unknown) underlying Gaussian densities are shown in yellow and the \DMDs{} in purple. (Right column) Comparison of the cumulative distributions, where the (unknown) underlying Gaussian distribution is shown in yellow and the distributions of the \DM s in purple.}
\label{Fig_Gaussian_Compare_LM_ME}
\end{figure*}

\begin{figure*}[t]
\includegraphics{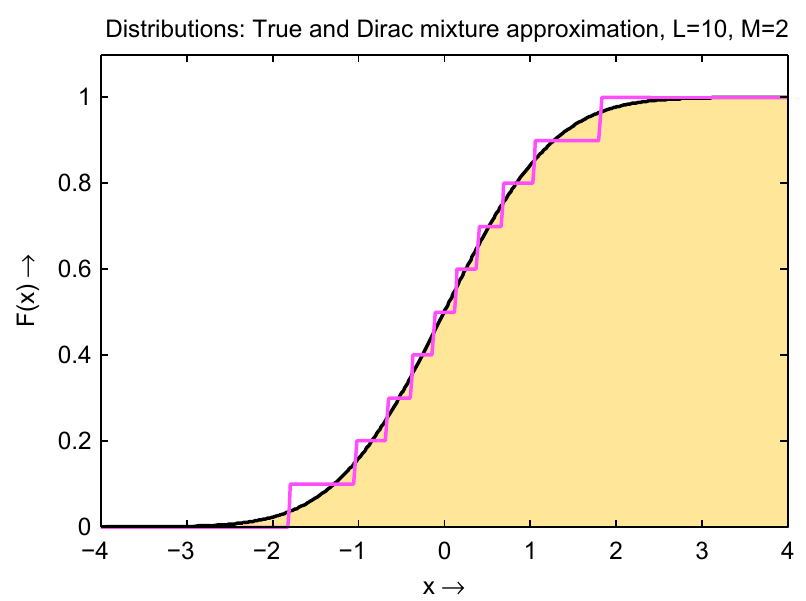}
\hspace*{\fill}
\includegraphics{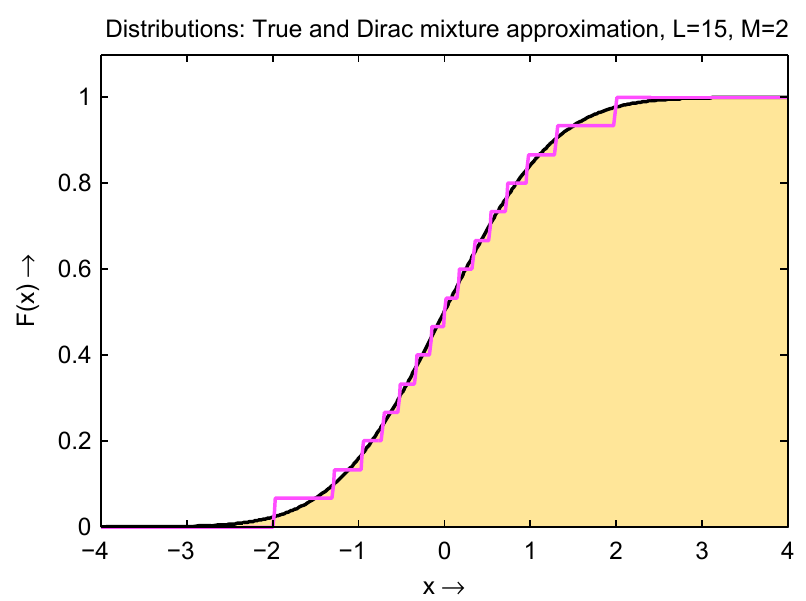}\\
\caption{\DMA{} for an underlying \GD{} with an increasing number of components $L$. Only moments up to order $M=2$ are maintained. (left) $L=10$. (Right) $L=15$. The \DMD{} quickly approaches the underlying Gaussian, although this density is not known to the approximation method.}
\label{Fig_Gaussian_Convergence}
\end{figure*}

%
%

A similar evaluation is now performed by using moments obtained from a \GMD{} with two components, weights $w_1=0.4$, $w_2=0.6$, means $m_1=-1.5$, $m_1=1.5$, and standard deviations $\sigma_1=\sigma_2=0.7$. Moments $e_0=1, e_1, \ldots, e_4$ up to fourth order are calculated according to the appendix and used for generating a \DM{} with these moments. \Fig{Fig_GM_Compare_LM_ME} shows the results for $L=10$ components, where the top row shows a comparison of densities and distributions for \DM s obtained with \LMDMA{}. Again, this is only one possible result as the solution is not unique. The bottom row shows the result obtained with \MEDMA{}, which is very homogeneous and close to the underlying \GMD{} in terms of its distribution.

\Fig{Fig_GM_Convergence} then shows that, for $M=6$ moments, the \MEDMA{} quickly converges to the underlying \GMD{} as demonstrated for $L=15$ and $L=25$.

\begin{figure*}[t]
\includegraphics{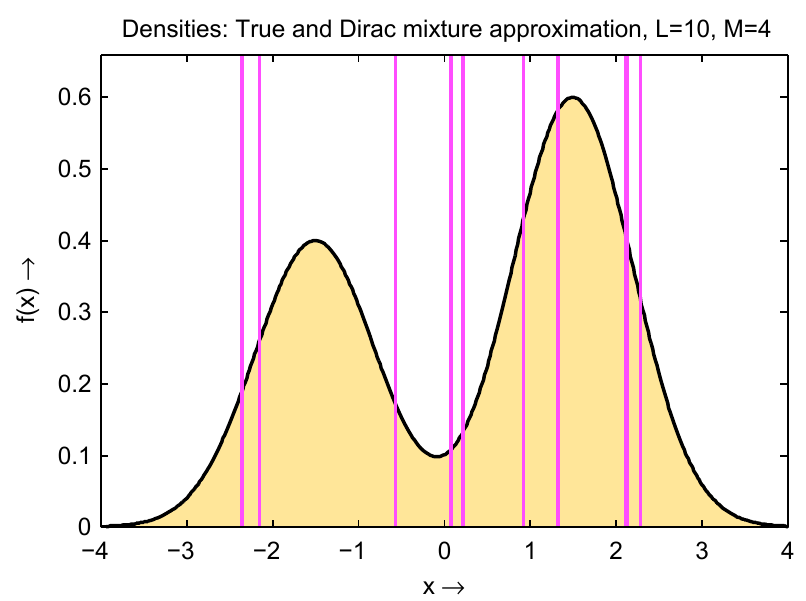}
\hspace*{\fill}
\includegraphics{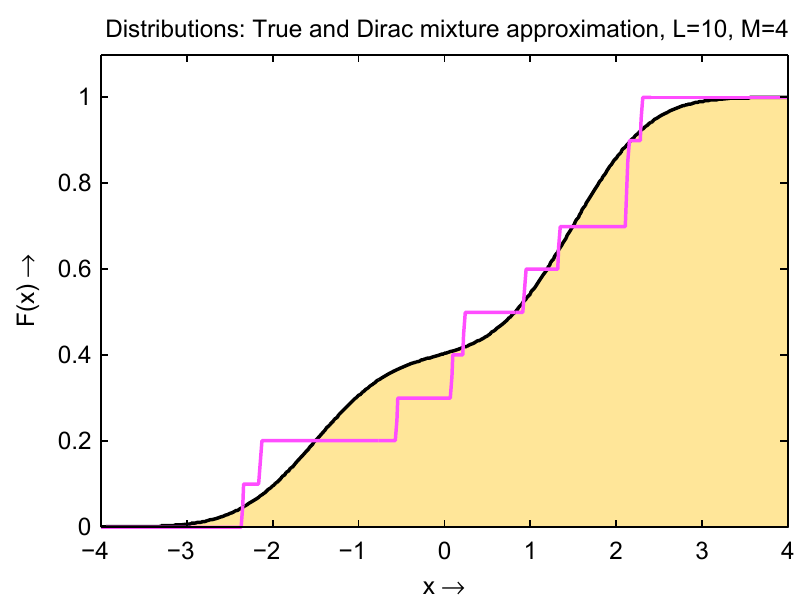}\\
\includegraphics{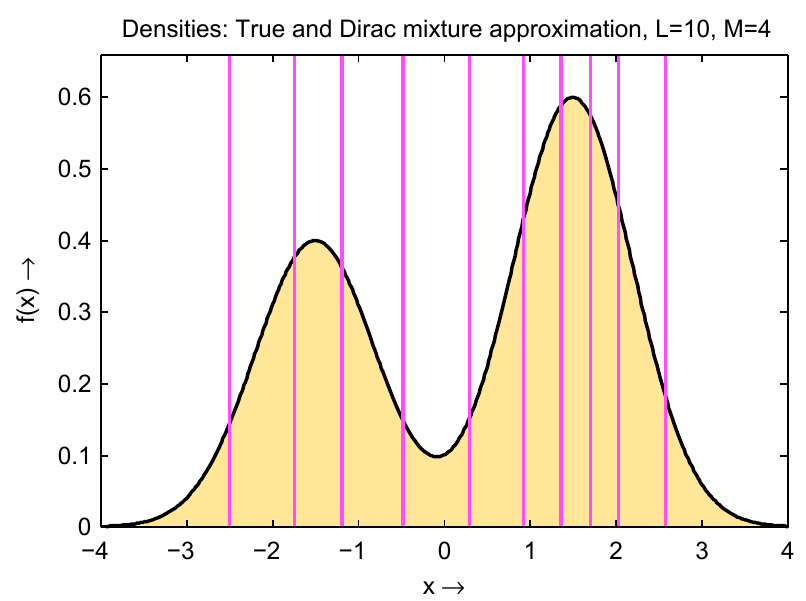}
\hspace*{\fill}
\includegraphics{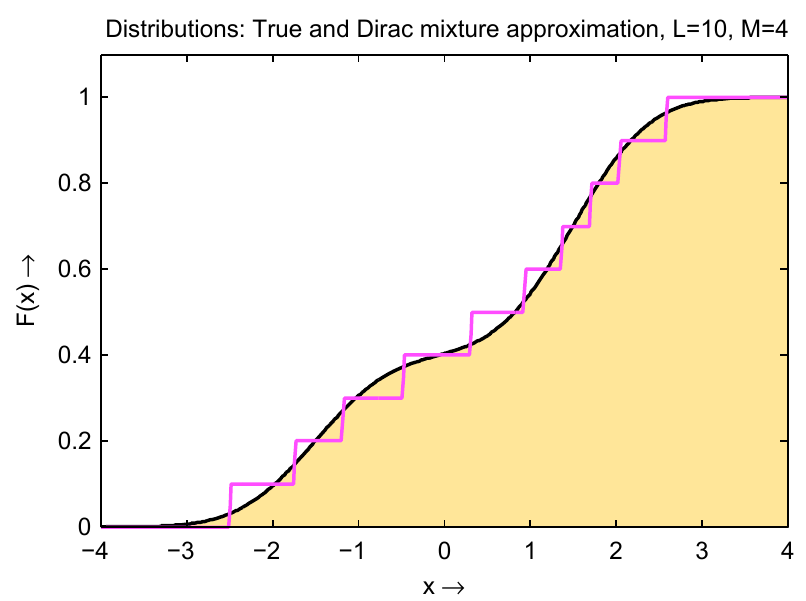}
\caption{Comparison of two \DMA{} methods for an underlying \GMD{}. (Top row) Result of solving the underdetermined moment constraints with the \LMDMA{} method. (Bottom row) Result of the proposed \MEDMA{} method. (Left column) Comparison of the densities, where the (unknown) underlying Gaussian mixture densities are shown in yellow and the \DMDs{} in purple. (Right column) Comparison of the cumulative distributions, where the (unknown) underlying Gaussian mixture distributions are shown in yellow and the distributions of the \DM s in purple. In both cases, moments up to order $M=4$ are maintained and the \DM{} comprises $L=10$ components.}
\label{Fig_GM_Compare_LM_ME}
\end{figure*}

\begin{figure*}[t]
\includegraphics{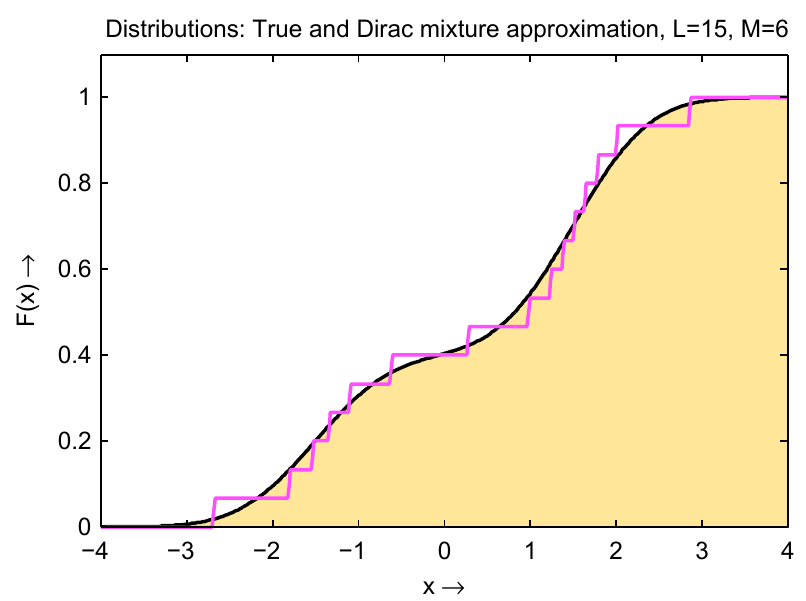}
\hspace*{\fill}
\includegraphics{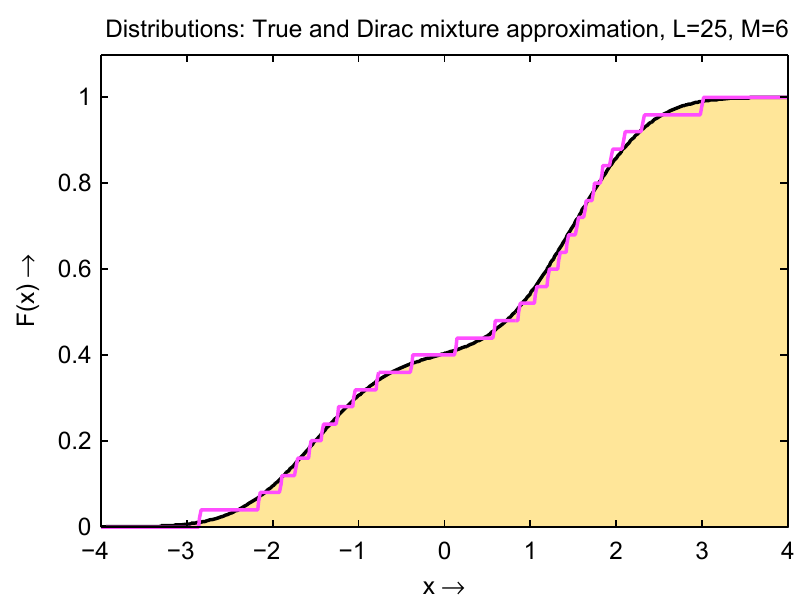}\\
\caption{\DMA{} for an underlying \GMD{} with an increasing number of components $L$. Moments up to order $M=6$ are maintained. (left) $L=15$. (Right) $L=25$. The \DMD{} quickly approaches the underlying Gaussian mixture, although this density is not know to the approximation method.}
\label{Fig_GM_Convergence}
\end{figure*}


\subsection{Examples for the Two-dimensional Case}

%
%

For the two-dimensional case $N=2$, we consider \DMDs{} maintaining moments up to second order, i.e., $e_{00}=1$ (the normalization constant), $e_{01}$, $e_{10}$, $e_{11}$, $e_{02}$, and $e_{20}$ are prespecified. For the specific choice of moments corresponding to the axis-aligned normal distribution $e_{00}=1$, $e_{01}=0$, $e_{10}=0$, $e_{11}=0$, $e_{02}=3$, and $e_{20}=1$, we obtain the results shown in \Fig{Fig_2D_Gaussian}. Here, symmetry is enforced as introduced in \SubSec{Subsection_Symmetry}. As a result, only $20$ master components are optimized and $20$ slave components follow accordingly.

\begin{figure*}[t]
\includegraphics{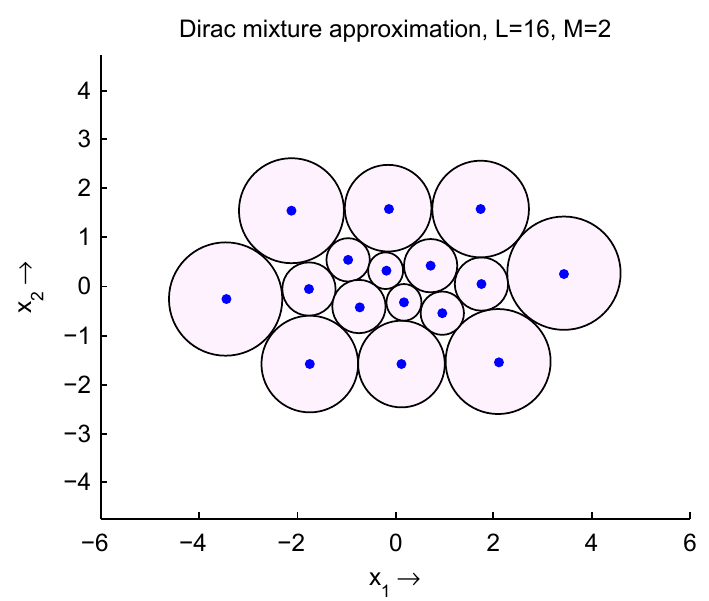}
\hspace*{\fill}
\includegraphics{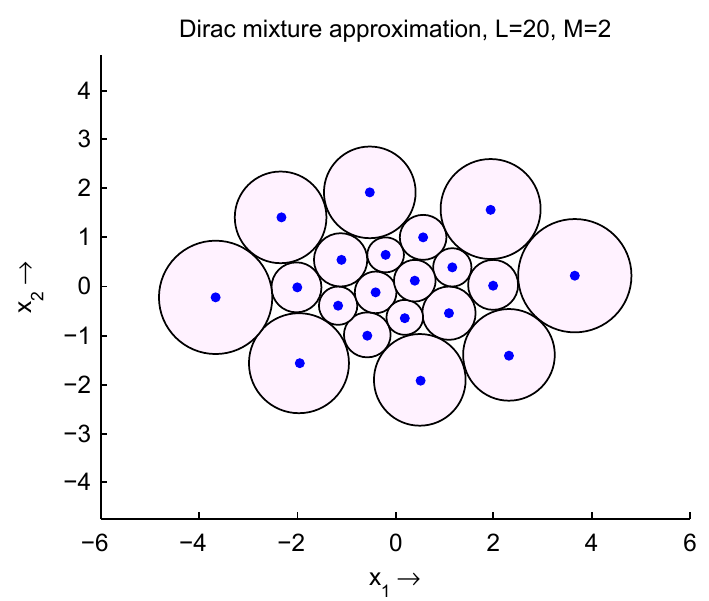}\\
\includegraphics{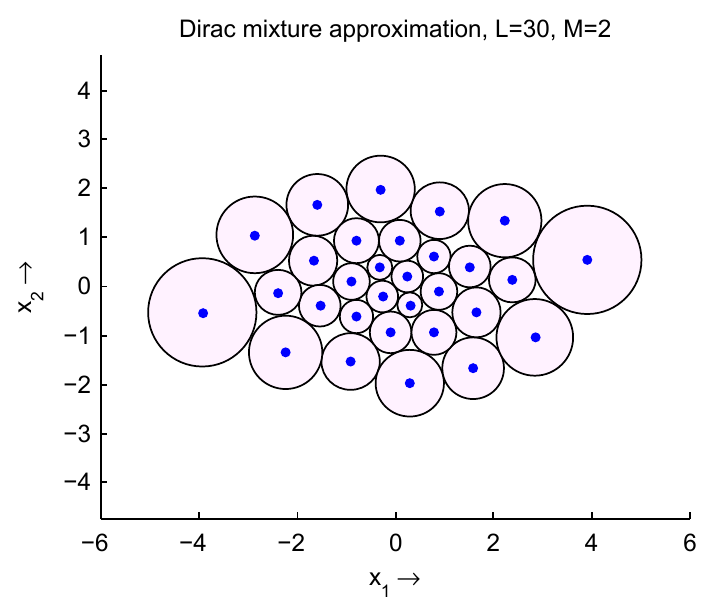}
\hspace*{\fill}
\includegraphics{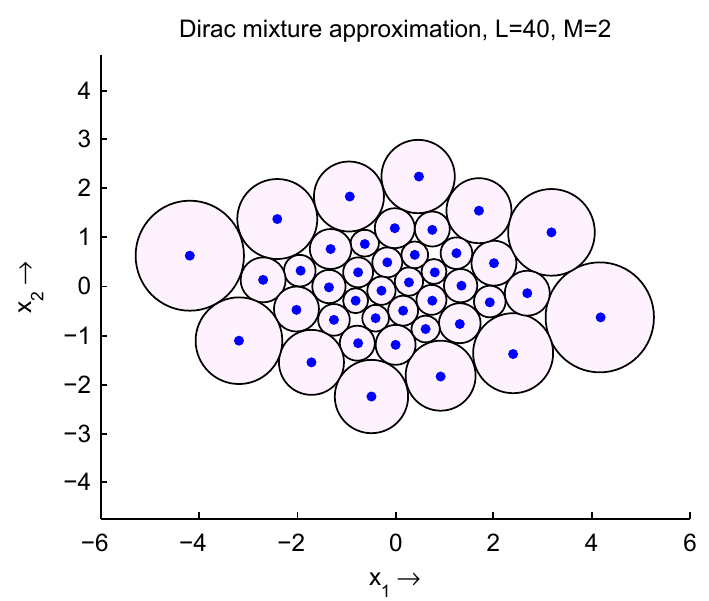}
\caption{Plots of \DMDs{} maintaining moments up to second order $M=2$ with different numbers of components. (top left) $L=16$. (top right) $L=20$. (bottom left) $L=30$. (bottom right) $L=40$. Symmetry is enforced, so that only $20$ master components are optimized.}
\label{Fig_2D_Gaussian}
\end{figure*}

%
%

It is obvious from \Fig{Fig_2D_Gaussian} that the resulting \DMDs{} converge to a Gaussian density with the given moments when the number of components increases. Again, the underlying density shape, in this case the Gaussian, is not known to the \DMA{} method.

\section{Conclusion} \label{Sec_Conclude}

%
%

This paper provides an efficient algorithm for calculating \DMDs{} with given moments and a homogeneous coverage of the state space. We focus on the case of fewer moment constraints than parameters. Ensuring a unique solution and exploiting the redundancy by optimizing the component arrangement is performed by regularization with respect to a corresponding \PWCD{}. The most homogeneous \DMD{} is obtained by maximizing the entropy of this \PWCD{}.

\subsection{Applications}

%
%

The proposed \MEDMA{} method will be used for generalizing the Progressive Gaussian Filter introduced for generative system models in \cite{Fusion13_Hanebeck} and for systems with given likelihoods in \cite{Fusion14_Steinbring}. So far, the Progressive Gaussian Filter maintains moments up to second order when progressively performing a measurement update. In addition, a Gaussian assumption is made. Using the proposed \MEDMA{}, higher-order moments will be propagated without any density assumption.

%
%

The \PWCD{} for a given \DMD{} as derived in \Sec{Sec_PiecewiseConstant} might also be useful by itself in other contexts than providing a convenient density for plug-in estimation of the entropy of a \DMD{}.

\subsection{Extensions}

%
%

Now, we will discuss several extensions to the basic algorithm described in this paper. First, we discuss optimizing the weights in addition to the locations of the components of the considered \DMD{}. Second, we focus on the complexity and how to decrease it. Third, we consider spaces different from the Euclidean space $\NewR^N$.


\paragraph{Weights}

%
%

In this paper, we focused on optimizing the locations of a \DMD{} only. Also optimizing the weights gives another $L-1$ degrees of freedom for $L$ components (as the weights have to sum up to one). This gives the advantage of potentially maintaining more moments for the same number of components.

%
%

The downside of optimizing the weights is obvious: Now there are components of different importance. Components with a small weight are almost negligible and do not carry much information although optimizing them is as costly as for large components. 
%
%
This is aggravated when components are fed through a nonlinear system in order to calculate the output density. In that case, the weights remain unchanged and the output weights are identical to the input weights. As the computational cost of propagating components does not depend on the weight, this implies that lots of computational power is spent for small output components with questionable usefulness.

%
%

The disadvantage of unequally weighted \DMDs{} becomes even more obvious when considering multiplication with a likelihood function in a Bayesian filter step. Already small components can potentially be weighted down even more making them useless, while for equally weighted components there is more leeway before components degenerate.


\paragraph{Complexity}

%
%

Let us briefly sum up the complexity of the algorithm described in this paper: For the scalar case, the complexity of the optimization problem is low. We have to maintain $M+1$ moment constraints (when considering moments up to order $M$) and $2 L - 1$ inequality constraints to optimize for $L$ location parameter of the desired \DMD{}. For the multi-dimensional case, we face two problems. First, the number of equality constraints corresponding to the number of moments up to order $M$ quickly increases with the number of dimensions $N$ so that calculating these moments for a \DMD{} soon becomes intractable. Second, the number of inequality constraints grows as $(L+1)L/2$ with the number of components.

%
%

For the multi-dimensional case, the first problem can be coped with by considering only the most relevant moments, which depends upon the application. The second problem, that is the quadratic growth of the inequality constraints, will be attacked by a two-step constraint hierarchy. The first step uses more conservative dimension-wise collision constraints, which results in $N (L-1)$ constraints. In the second step, multi-dimensional constraints are only assembled for the remaining \DC s that need further attention. This is expected to result in an additional ${\cal O}(L)$, say $k \, L$, constraints. Together with the additional positivity constraint for the radii $d_i$, $i=1,\ldots, L$, we obtain a number of about $k \, L + N(L-1) + L$ constraints. Hence, we trade an algorithm with a total number of $(L+1)L/2$ constraints for an algorithm with a number of $k \, L + N(L-1) + L$ constraints, which is more efficient when $N<L(L-1-2 k)/(2(L-1))$. This is already the case for $k=1$, $N=2$, and $L=7$.
%
%
Implementation of these hierarchical constraints, however, is a challenge as the number of constraints varies during runtime of the optimization procedure. Available optimization routines do not seem be able to cope with a varying number of constraints.


\paragraph{Alternative Spaces}

The techniques presented in this paper for the case that no underlying density is available, will be generalized to different Polish spaces ${\cal X}$, especially to periodic spaces such as the unit circle $S^1$ as it has been done in \cite{IFAC14_Hanebeck} for known \cpdf s.


%

\bibliographystyle{StyleFiles/IEEEtran_Capitalize}

\bibliography{%
Literature,%
ISAS-Bibtex/ISASPublikationen,%
ISAS-Bibtex/ISASPublikationen_laufend%
}

\begin{appendix}

\section{Moments of Scalar Gaussian Density}

When mean and standard deviation of a Gaussian random variable are
given, \emph{all} higher--order moments and central moments can be
deduced analytically. The central moments are given by
\begin{equation*}
C_i = E_{f} \left\{ (x-m)^i \right\} =
  \begin{cases}
    \displaystyle\prod_{\stackrel{j=1}{j \text{ odd}}}^{i-1} j \, \sigma^i & i \text{ even} \enspace , \\[2mm]
    0 & i \text{ odd} \enspace ,
  \end{cases}
\end{equation*}
the moments by
\begin{equation*}
E_i = E_{f} \left\{ x^i \right\} = \sum_{k=0}^i {i \choose k} C_{i-k} \, m^k
\end{equation*}
with $C_0$ (the zeroth central moment, the area under the density) is defined as $C_0=1$.

\begin{example}[First Moments and Central Moments of Gaussian Density]
The first eight moments $E_i$, $i=1,\ldots,8$, and central moments
$C_i$, $i=1,\ldots,8$, of a Gaussian density with mean $m$ and
standard deviation $\sigma$ are given by
\begin{equation*}
\begin{split}
C_1 & = 0 \enspace , \\
C_2 & = \sigma^2 \enspace , \\
C_3 & = 0 \enspace , \\
C_4 & = 3 \, \sigma^4 \enspace , \\
C_5 & = 0 \enspace , \\
C_6 & = 15 \, \sigma^6 \enspace , \\
C_7 & = 0 \enspace , \\
C_8 & = 105 \, \sigma^8 \enspace ,
\end{split}
\,
\begin{split}
E_1 & = m \enspace , \\
E_2 & = m^2 + \sigma^2 \enspace , \\
E_3 & = m^3 + 3 \, m \, \sigma^2 \enspace , \\
E_4 & = m^4 + 6 \, m^2 \, \sigma^2 + 3 \, \sigma^4 \enspace , \\
E_5 & = m^5 + 10 \, m^3 \, \sigma^2 + 15 \, m \, \sigma^4 \enspace , \\
E_6 & = m^6 + 15 \, m^4 \, \sigma^2 + 45 \, m^2 \, \sigma^4 + 15 \, \sigma^6 \enspace , \\
E_7 & = m^7 + 21 \, m^5 \, \sigma^2 + 105 \, m^3 \, \sigma^4 + 105 \, m \, \sigma^6 \enspace , \\
E_8 & = m^8 + 28 \, m^6 \, \sigma^2 + 210 \, m^4 \, \sigma^4  + 420\, m^2 \, \sigma^6 + 105 \, \sigma^8.
\end{split}
\end{equation*}
\end{example}


\section{Moments of Mixture}

We consider scalar mixtures of the form
\begin{equation*}
f(x) = \sum_{k=1}^P w_k \, f_k(x)
\end{equation*}
with
\begin{equation*}
w_k \ge 0
\end{equation*}
for $k=1, \ldots, P$ and 
\begin{equation*}
\sum_{k=1}^P w_k = 1 \enspace .
\end{equation*}
When the moments of the \emph{individual} densities $f_k(.)$ in the mixture are given by $E_i^{(k)}$ for $k=1,\ldots,P$, the individual moments can now be added up to the moments of the mixture denoted by $E_i^{M}$
\begin{equation*}
E_i^{M} = \sum_{k=1}^P w_k \, E_i^{(k)}
\enspace .
\end{equation*}
Finally, central moments of the mixture are obtained as
\begin{equation*}
C_i^{M} = \sum_{j=0}^i {i \choose j} E_{i-j}^{M} \, \left(
-E_1^{M} \right)^j
\enspace .
\end{equation*}

\end{appendix}

\end{document}